\documentclass[preprint,12pt]{elsarticle}

\usepackage[a4paper,margin=1in]{geometry}
\usepackage{amsmath,amssymb,amsthm}
\usepackage{booktabs}
\usepackage{graphicx}
\usepackage[hidelinks]{hyperref}
\usepackage{microtype}
\usepackage{algorithm}
\usepackage{algpseudocode}
\usepackage{tikz}
\usetikzlibrary{positioning}

\newtheorem{definition}{Definition}
\newtheorem{lemma}{Lemma}
\newtheorem{proposition}{Proposition}
\newtheorem{theorem}{Theorem}

\begin{document}

\begin{frontmatter}
\title{An Efficient Algorithm for the Quickest Path Reliability Problem}

\author[inst1]{Majid Forghani-elahabad}

\affiliation[inst1]{organization={Center of Mathematics, Computing, and Cognition - Federal University of ABC},
            city={Santo André},
            postcode={09280-560}, 
            state={SP},
            country={Brazil}}

\date{}

\begin{abstract}
The quickest path reliability problem is to compute the probability that a given flow is transmitted along a single source-destination path within a given time limit. When a budget is imposed, path feasibility is jointly dependent on lead time, transmission cost, and capacity. Existing exact methods may require enumeration of all minimal paths a priori or examination of many partial paths that cannot be completed feasibly. In this paper, we propose an exact algorithm that generates the corresponding lower boundary state vectors without first enumerating minimal paths. Shortest-lead-time, least-cost, and widest-capacity information yield admissible destination-based bounds. Stronger capacity-aware bounds are derived from capacity-filtered subnetworks. These bounds allow early pruning of infeasible branches while keeping all solutions. The algorithm generates the exact complete set of lower bounds without any duplicates, and its correctness and computational complexity are proven. 
\end{abstract}


\end{frontmatter}

\section{Introduction}

\section{Introduction}
\label{sec:introduction}

Many systems encountered in practice can be represented as flow networks, in which a resource must be transmitted from an origin to a destination through a collection of interconnected components~\cite{ahuja1993network, forghani2015ieee}. Depending on the application, the transmitted resource may represent data, commodities, vehicles, electrical energy, or production materials. Flow-network models have therefore been employed in communication and computer networks~\cite{PANG2026112670,yeh2021novel, forghani2021chapter}, transportation and logistics systems~\cite{niu2024expected,xiahou2023reliability}, power transmission and distribution systems~\cite{lin2022reliability, forghani2015ress}, manufacturing processes~\cite{hao2020general,huang2024manufacturing}, and supply chains~\cite{xu2022multistate}. In such systems, component capacities are rarely fixed permanently. Failures, maintenance, congestion, environmental conditions, and operational decisions may cause them to assume different values over time~\cite{forghani2023usage,niu2017new,kozyra2024parallel, forghani2020assa}.

A multistate flow network (MFN) represents this uncertainty by allowing each arc to operate at several capacity levels, each occurring according to a specified probability distribution~\cite{niu2020finding-MP,forghani2019ress,SAMIE2027113089}. One of the main measures used to assess the performance of an MFN is its two-terminal reliability. In its basic form, this measure is the probability that a required amount of flow can be transmitted from a designated source node to a designated destination node~\cite{shier1991network}. In many applications, however, the amount that can be transmitted is not the only relevant consideration. The transmission must also be completed within an acceptable period of time~\cite{forghani20223, forghani2014ieee}.

For a deterministic network, this requirement leads to the quickest path problem. Given a demand, an arc capacity, and an arc lead time, the objective is to find a source--destination path through which the demand can be transmitted in the shortest total time~\cite{moore1976fastest,chen1990quickest}. The problem differs from the ordinary shortest path problem because the transmission time depends on two characteristics of the selected path: the sum of its arc lead times and its bottleneck capacity. Consequently, the path having the smallest total lead time need not be the path that completes the transmission first. This particular structure has motivated several solution procedures, including capacity-filtered subnetworks, transformed-network methods, and constrained path formulations~\cite{calvete2012algorithms,ghiani2014lower,sedeno2014fast}.

When the arc capacities are stochastic, the transmission time becomes a random quantity as well~\cite{HUANG2027113228}. Lin~\cite{lin2003extend} introduced the quickest path reliability problem in an MFN and defined the corresponding reliability as the probability that a prescribed demand can be transmitted through a single minimal path within a specified time limit. The proposed exact procedure begins by enumerating all minimal paths. For each path that can satisfy the time requirement, it determines the minimum necessary bottleneck capacity and constructs the associated lower boundary state vector. The required reliability is then obtained from the probability of the union of the corresponding upper-state events. This formulation established an important connection between the deterministic quickest path problem and multistate network reliability. Its computational burden, however, may be considerable because the number of minimal paths can grow exponentially with the network size~\cite{forghani20223, yeh2024efficient-bdcc, forghani2017ieee}.

Several later studies sought to reduce or avoid the need for preliminary minimal-path enumeration. Yeh et al.~\cite{yeh2009simpleQPR} used a $k$-shortest-path procedure together with subnetworks induced by different capacity levels. Yeh and El Khadiri~\cite{yeh2012new} subsequently introduced an algebraic solution based on universal generating functions. A depth-first-search approach was later proposed by Yeh~\cite{yeh2015fast}. In that method, source--destination paths are constructed incrementally, and a partial path is abandoned when its accumulated lead time is already too large to satisfy the prescribed time limit. These procedures represented an important development because they no longer required the complete minimal-path set to be supplied beforehand. Nevertheless, a search may still explore many partial paths whose infeasibility becomes apparent only after they have been extended close to the destination.

The construction of the lower boundary vectors is only one part of exact quickest path reliability evaluation. Once these vectors have been obtained, the associated upper-state events generally overlap. Their union probability must therefore be calculated by an exact reliability-evaluation method, such as inclusion--exclusion, the sum of disjoint products, or state-space decomposition~\cite{balan2003preprocessing,schafer2018simplification, mansourzadeh2014comparative}. El Khadiri et al.~\cite{elkhadiri2023factoring} proposed an alternative exact method based on recursive factoring and the total expectation theorem. This approach avoids the explicit generation of all relevant paths and can be advantageous when a large collection of overlapping path events would otherwise have to be processed.

Approximation and simulation procedures provide another possibility when exact probability evaluation becomes excessively demanding. El Khadiri and Yeh~\cite{el2016efficient} developed an unbiased Monte Carlo estimator for quickest path reliability. More recent procedures estimate the reliability by sampling network-capacity states and have also considered extensions in which the arc lead times are stochastic~\cite{huang2026montecarlo,huang2026stochastic}. Such methods are useful for large-scale problems, especially when an approximate reliability value is sufficient. The present study has a different objective: it focuses on the exact generation of the complete lower boundary set associated with a prescribed demand, time limit, and budget.

The inclusion of transmission cost makes the problem more representative of applications in which time and expenditure must be controlled simultaneously. Lin~\cite{lin2010system} extended the quickest path reliability model by imposing a budget constraint and constructed the lower boundary vectors associated with paths satisfying both restrictions. Forghani-elahabad and Yeh~\cite{forghani2022improvedress} later showed that the time and budget conditions can be applied earlier to exclude unsuitable minimal paths before their state vectors are formed. Despite these improvements, both procedures require the complete collection of minimal paths as an input.

To eliminate this requirement, Forghani-elahabad and Alsalami~\cite{forghani2023using} developed a backtracking method based on the node--child matrix. Instead of enumerating all minimal paths in a separate preliminary phase, their algorithm constructs simple paths progressively while monitoring accumulated lead time, accumulated transmission cost, and bottleneck capacity. Infeasible partial paths can consequently be eliminated during the search. This method provides a suitable basis for the direct generation of lower boundary vectors, but its pruning decisions rely mainly on the resources already consumed by the current partial path. They do not fully account for the additional lead time and cost that will necessarily be required to reach the destination. Hence, a partial path can satisfy all current tests even when none of its possible completions is feasible.

A different recent contribution was made by Kozyra~\cite{kozyra2026dictionary}, who proposed a dictionary method for a generalized model with possibly noninteger capacities and correlated component states. For a fixed demand and budget, the method efficiently identifies the attainable transmission times and their corresponding lower boundary vectors. It is particularly useful when results are required for several time limits. However, the complete set of minimal paths must be available before the method is applied. Thus, when a single prescribed combination of demand, time limit, and budget is considered, improving the direct construction of feasible paths without preliminary minimal-path enumeration remains a relevant problem.

The central idea of this paper is to make the search look beyond the resources already consumed. For each node, we determine optimistic destination-based values for the additional lead time and cost required to complete a partial path. These values are used as admissible lower bounds rather than as prescribed continuations of that path. A partial path can therefore be discarded as soon as its accumulated resource consumption, combined with the least resource consumption still required, exceeds one of the given limits.

This idea is further strengthened by computing the destination-based bounds in subnetworks induced by different capacity levels. The resulting tests account jointly for the additive lead-time and cost restrictions and for the bottleneck capacity governing the transmission time. If none of the remaining capacity levels can provide a feasible completion, the current partial path can be safely eliminated. In this manner, infeasibility may be detected considerably earlier in the path-construction process without excluding any lower boundary vector.

The principal contributions of the paper can be summarized as follows:
\begin{itemize}
    \item Destination-based lower bounds are derived for the additional lead time and transmission cost needed to extend a partial path to the destination.
    
    \item The bounds are strengthened by means of capacity-filtered subnetworks, leading to a capacity-aware feasibility test that captures the interaction among remaining lead time, remaining budget, and attainable bottleneck capacity.
    
    \item The resulting conditions are incorporated into an exact resource-aware look-ahead algorithm that directly generates the complete set of $(d,T,b)$ lower boundary vectors, without requiring prior enumeration of all minimal paths and without producing duplicate solutions.
    
    \item The safety of the proposed pruning rules and the soundness, completeness, and correctness of the algorithm are established formally. Output-sensitive time and memory analyses are also developed by distinguishing the costs of preprocessing, path enumeration, and explicit solution construction.
\end{itemize}

The remainder of the paper is organized as follows. Section~\ref{sec:preliminaries} presents the network model, assumptions, notation, and preliminary results required for the subsequent developments. Section~\ref{sec:algorithm} derives the destination-based and capacity-aware conditions, describes the proposed resource-aware look-ahead algorithm, and analyzes its theoretical properties and computational complexity. Section~\ref{sec:example} illustrates the successive steps of the algorithm using a benchmark network. Finally, Section~\ref{sec:conclusion} summarizes the main findings and discusses possible directions for future research.

\section{Problem Formulation and Preliminaries}
\label{sec:preliminaries}

This section introduces the network model, assumptions, and notation used throughout the paper. We then formulate the time- and budget-constrained quickest path reliability problem, characterize its lower boundary state vectors, and briefly explain how the network reliability can be calculated once these vectors are available.

\subsection{Network Model and Notation}
\label{subsec:network_model}

Consider a multistate flow network represented by
\[
G=(N,A,\boldsymbol{M},\boldsymbol{\ell},\boldsymbol{c}),
\]
where \(N=\{1,2,\ldots,n\}\) is the node set and
\(A=\{a_1,a_2,\ldots,a_m\}\) is the arc set. Nodes \(1\) and \(n\) are respectively designated as the source and destination nodes. The network may be directed or undirected. In the undirected case, an arc can be traversed in either direction, subject to the condition that a path cannot visit the same node more than once.

The maximum-capacity vector is denoted by
\[
\boldsymbol{M}=(M_1,M_2,\ldots,M_m),
\]
where \(M_i\) is the maximum capacity of arc \(a_i\). The deterministic lead-time and unit-cost vectors are respectively denoted by
\[
\boldsymbol{\ell}=(\ell_1,\ell_2,\ldots,\ell_m)
\quad\text{and}\quad
\boldsymbol{c}=(c_1,c_2,\ldots,c_m),
\]
where \(\ell_i\) is the lead time of arc \(a_i\), and \(c_i\) is the cost of transmitting one unit of flow through that arc.

The current capacity of arc \(a_i\) is represented by the random variable \(x_i\), whose possible values belong to
\[
\mathcal{S}_i=\{0,1,\ldots,M_i\}.
\]
The system state vector is consequently
\[
\boldsymbol{x}=(x_1,x_2,\ldots,x_m),
\qquad
0\leq x_i\leq M_i.
\]
Let
\[
p_{i,k}=\Pr(x_i=k),
\qquad
k=0,1,\ldots,M_i,
\]
denote the probability that arc \(a_i\) operates at capacity level \(k\).

For two state vectors \(\boldsymbol{x}\) and \(\boldsymbol{y}\), the componentwise relation
\[
\boldsymbol{x}\leq\boldsymbol{y}
\]
means that \(x_i\leq y_i\) for every \(i=1,\ldots,m\). Moreover,
\(\boldsymbol{x}<\boldsymbol{y}\) means that
\(\boldsymbol{x}\leq\boldsymbol{y}\) and
\(\boldsymbol{x}\neq\boldsymbol{y}\). The componentwise maximum of the two vectors is denoted by
\[
\boldsymbol{x}\vee\boldsymbol{y}
=
\bigl(\max\{x_1,y_1\},\ldots,\max\{x_m,y_m\}\bigr).
\]
This notation is extended naturally to any finite collection of state vectors.

A source--destination path \(P\) is a sequence of adjacent arcs connecting node \(1\) to node \(n\). Whenever no confusion can arise, the same symbol is used for a path and for its associated ordered arc set. A path is called a minimal path if no proper subset of its arcs still contains a path from node \(1\) to node \(n\). The set of all minimal paths (MPs) is denoted by
\[
\mathcal{P}=\{P_1,P_2,\ldots,P_h\},
\]
where \(h=|\mathcal{P}|\) is the number of MPs. The proposed algorithm will not require \(\mathcal{P}\) to be generated in advance.

For a path \(P\), its lead time and unit transmission cost are respectively defined as
\begin{equation}
L(P)=\sum_{a_i\in P}\ell_i
\label{eq:path_lead_time}
\end{equation}
and
\begin{equation}
C(P)=\sum_{a_i\in P}c_i.
\label{eq:path_unit_cost}
\end{equation}
Under the system state \(\boldsymbol{x}\), the available capacity of \(P\) is its bottleneck capacity:
\begin{equation}
K(P,\boldsymbol{x})
=
\min_{a_i\in P}x_i.
\label{eq:path_capacity}
\end{equation}

The main parameters of the problem are the demand \(d\), the time limit \(T\), and the available budget \(b\). In particular, \(d\) is the amount of flow that must be transmitted from node \(1\) to node \(n\), \(T\) is the maximum allowable transmission time, and \(b\) is the maximum allowable transmission cost. The demand and time limit are assumed to be positive integers, and the available budget is nonnegative.

\subsection{Assumptions}
\label{subsec:assumptions}

The following assumptions, commonly adopted in the quickest path reliability literature, are considered throughout this paper
\cite{lin2003extend,lin2010system,forghani2022improvedress,forghani2023using}.

\begin{enumerate}
    \item The capacity \(X_i\) of every arc \(a_i\) is a discrete random variable taking an integer value in
    \(\{0,1,\ldots,M_i\}\), and its probability distribution is known.

    \item The capacities of different arcs are statistically independent.


    \item The maximum capacity \(M_i\) is a positive integer, the lead time \(\ell_i\) is a nonnegative integer, and the unit transmission cost \(c_i\) is nonnegative for every arc \(a_i\in A\). These quantities are known and deterministic.

    \item The network satisfies the flow-conservation law. Hence, no intermediate node generates, consumes, or stores flow.

    \item All \(d\) units of flow are transmitted through a single source--destination path. Flow splitting among multiple paths is not permitted.

    \item Every node is perfectly reliable. An unreliable node, if present, can be replaced by two perfectly reliable nodes connected by an arc whose capacity distribution represents the state of the original node~\cite{forghani2016amm}.


\end{enumerate}

The independence assumption is needed for the probability calculations presented in Subsection~\ref{subsec:reliability_boundary}. In contrast, the boundary-vector generation procedure developed later depends only on the possible arc-capacity levels and remains valid independently of their probabilities.

\subsection{Transmission Time and Feasible Paths}
\label{subsec:feasible_paths}

Suppose that \(d\) units of flow are transmitted through a path \(P\) under system state \(\boldsymbol{x}\). If \(K(P,\boldsymbol{x})=0\), the path cannot transmit any positive amount of flow. Otherwise, the total cost of sending the demand through \(P\) is
\begin{equation}
g(d,P)=dC(P)
      =d\sum_{a_i\in P}c_i.
\label{eq:transmission_cost}
\end{equation}
Thus, under the adopted cost model, the path capacity does not affect the transmission cost.

The transmission time, however, depends on both the path lead time and its current capacity. More precisely, the time required to send \(d\) units of flow through \(P\) under state \(\boldsymbol{x}\) is
\begin{equation}
f(d,\boldsymbol{x},P)
=
\begin{cases}
L(P)+\displaystyle\left\lceil
\frac{d}{K(P,\boldsymbol{x})}
\right\rceil,
& K(P,\boldsymbol{x})>0,\\[3mm]
+\infty,
& K(P,\boldsymbol{x})=0,
\end{cases}
\label{eq:transmission_time}
\end{equation}
where \(\lceil \bullet \rceil\) is the smallest integer number not less than \(\bullet\). The first term in \eqref{eq:transmission_time} is the time required for the first transmitted units to traverse the path, whereas the second term represents the number of transmission batches required when at most \(K(P,\boldsymbol{x})\) units can be sent at a time. Consequently, a path with the smallest lead time is not necessarily the path with the smallest transmission time.

\begin{definition}
\label{def:feasible_path}
For a given demand \(d\), time limit \(T\), budget \(b\), and system state \(\boldsymbol{x}\), a source--destination path \(P\) is called \emph{\((d,T,b)\)-feasible} if
\begin{equation}
g(d,P)\leq b
\quad\text{and}\quad
f(d,\boldsymbol{x},P)\leq T.
\label{eq:path_feasibility}
\end{equation}
\end{definition}

For a path \(P\) satisfying \(L(P)<T\), let
\begin{equation}
q(P)=T-L(P)
\label{eq:remaining_transmission_time}
\end{equation}
denote the time available for transmitting the demand after accounting for the lead time. The minimum path capacity required to complete the transmission within \(T\) is then
\begin{equation}
\alpha(P)
=
\left\lceil
\frac{d}{T-L(P)}
\right\rceil.
\label{eq:required_path_capacity}
\end{equation}

\begin{proposition}
\label{prop:path_feasibility}
Let \(P\) be a source--destination path. The path \(P\) is
\((d,T,b)\)-feasible under state \(\boldsymbol{x}\) if and only if
\begin{equation}
L(P)<T,\qquad
dC(P)\leq b,\qquad
K(P,\boldsymbol{x})\geq\alpha(P).
\label{eq:equivalent_feasibility}
\end{equation}
\end{proposition}

\begin{proof}
If \(L(P)\geq T\), no positive amount of flow can traverse \(P\) within the time limit. Suppose, therefore, that \(L(P)<T\). From \eqref{eq:transmission_time},
\[
f(d,\boldsymbol{x},P)\leq T
\]
is equivalent to
\[
\left\lceil
\frac{d}{K(P,\boldsymbol{x})}
\right\rceil
\leq T-L(P).
\]
Because \(T-L(P)\) is a positive integer, the preceding inequality holds if and only if
\[
K(P,\boldsymbol{x})
\geq
\left\lceil
\frac{d}{T-L(P)}
\right\rceil
=
\alpha(P).
\]
Combining this condition with \(g(d,P)=dC(P)\leq b\) proves the result.
\end{proof}

Proposition~\ref{prop:path_feasibility} is particularly useful because it separates the deterministic attributes of a path from the random system state. The quantities \(L(P)\), \(C(P)\), and \(\alpha(P)\) depend only on the path and the prescribed parameters \(d\), \(T\), and \(b\), whereas the condition
\(K(P,\boldsymbol{x})\geq\alpha(P)\) depends on the current arc capacities.

The maximum capacity that can possibly be provided by \(P\) is
\begin{equation}
M(P)=\min_{a_i\in P}M_i.
\label{eq:path_max_capacity}
\end{equation}
Therefore, a path can be feasible under at least one system state only if
\begin{equation}
L(P)<T,\qquad
dC(P)\leq b,\qquad
\alpha(P)\leq M(P).
\label{eq:potential_path_feasibility}
\end{equation}
A path satisfying \eqref{eq:potential_path_feasibility} will be called a
\emph{potentially \((d,T,b)\)-feasible path}.

For a fixed system state \(\boldsymbol{x}\), define
\begin{equation}
F(d,\boldsymbol{x},b)
=
\min
\left\{
f(d,\boldsymbol{x},P):
P\in\mathcal{P},\
K(P,\boldsymbol{x})>0,\
dC(P)\leq b
\right\},
\label{eq:minimum_transmission_time}
\end{equation}
where the minimum of an empty set is defined as \(+\infty\). Hence,
\(F(d,\boldsymbol{x},b)\leq T\) if and only if at least one path can transmit the prescribed demand within both the time and budget limits.

Let
\begin{equation}
\Psi_{d,T,b}
=
\left\{
\boldsymbol{x}:
\boldsymbol{0}\leq\boldsymbol{x}\leq\boldsymbol{M},
\quad
F(d,\boldsymbol{x},b)\leq T
\right\}
\label{eq:successful_state_set}
\end{equation}
be the set of all successful system states. The quickest path reliability under the budget constraint is consequently
\begin{equation}
R_{d,T,b}
=
\Pr\bigl(\boldsymbol{X}\in\Psi_{d,T,b}\bigr).
\label{eq:qpr_definition}
\end{equation}

The success property is monotone with respect to the arc capacities, as stated next.

\begin{lemma}
\label{lem:monotonicity}
If \(\boldsymbol{x}\leq\boldsymbol{y}\), then
\[
F(d,\boldsymbol{x},b)
\geq
F(d,\boldsymbol{y},b).
\]
Consequently, if \(\boldsymbol{x}\in\Psi_{d,T,b}\) and
\(\boldsymbol{y}\geq\boldsymbol{x}\), then
\(\boldsymbol{y}\in\Psi_{d,T,b}\).
\end{lemma}

\begin{proof}
For every path \(P\),
\[
K(P,\boldsymbol{x})
\leq
K(P,\boldsymbol{y}).
\]
It follows from \eqref{eq:transmission_time} that
\[
f(d,\boldsymbol{x},P)
\geq
f(d,\boldsymbol{y},P).
\]
Moreover, the path cost does not depend on the system state. Taking the minimum over all paths satisfying the budget constraint gives the stated inequality.
\end{proof}

\subsection{Boundary State Vectors}
\label{subsec:boundary_vectors}

Since \(\Psi_{d,T,b}\) is an upper set under the componentwise ordering, it can be represented completely by its minimal elements.

\begin{definition}
\label{def:boundary_vector}
A state vector \(\boldsymbol{x}\) is called a
\emph{\((d,T,b)\) lower boundary vector} if
\[
\boldsymbol{x}\in\Psi_{d,T,b}
\]
and there is no state vector
\(\boldsymbol{y}\in\Psi_{d,T,b}\) such that
\(\boldsymbol{y}<\boldsymbol{x}\).
\end{definition}

The complete set of \((d,T,b)\) lower boundary vectors is denoted by
\begin{equation}
\Psi^{\min}_{d,T,b}
=
\left\{
\boldsymbol{x}^{1},
\boldsymbol{x}^{2},
\ldots,
\boldsymbol{x}^{\sigma}
\right\},
\label{eq:boundary_set}
\end{equation}
where \(\sigma=|\Psi^{\min}_{d,T,b}|\).

For every potentially feasible minimal path \(P\), define its associated state vector
\(\boldsymbol{x}^{P}\) by
\begin{equation}
x_i^{P}
=
\begin{cases}
\alpha(P), & a_i\in P,\\
0,         & a_i\notin P.
\end{cases}
\label{eq:path_boundary_vector}
\end{equation}
Because \(\alpha(P)\leq M(P)\), the vector
\(\boldsymbol{x}^{P}\) is a valid system state. Furthermore,
\[
K(P,\boldsymbol{x}^{P})=\alpha(P),
\]
and Proposition~\ref{prop:path_feasibility} guarantees that
\(\boldsymbol{x}^{P}\in\Psi_{d,T,b}\).

\begin{proposition}
\label{prop:path_to_boundary}
If \(P\) is a potentially \((d,T,b)\)-feasible minimal path, then the vector \(\boldsymbol{x}^{P}\) defined by
\eqref{eq:path_boundary_vector} is a \((d,T,b)\) lower boundary vector.
\end{proposition}

\begin{proof}
The vector \(\boldsymbol{x}^{P}\) is successful because \(P\) has capacity
\(\alpha(P)\) under this state and satisfies the lead-time and budget conditions. Suppose that a component of
\(\boldsymbol{x}^{P}\) is reduced. Then, at least one arc of \(P\) has capacity smaller than \(\alpha(P)\), and \(P\) can no longer satisfy the time limit.

In addition, all arcs outside \(P\) have capacity zero in
\(\boldsymbol{x}^{P}\). Hence, any other successful source--destination path under a state smaller than \(\boldsymbol{x}^{P}\) would have to be composed entirely of arcs belonging to \(P\). This would imply the existence of a proper subset of \(P\) that is also a source--destination path, contradicting the minimality of \(P\). Therefore, no successful state is strictly smaller than
\(\boldsymbol{x}^{P}\).
\end{proof}

\begin{proposition}
\label{prop:boundary_to_path}
For every \((d,T,b)\) lower boundary vector
\(\boldsymbol{x}\), there exists a potentially feasible minimal path \(P\) such that
\[
\boldsymbol{x}=\boldsymbol{x}^{P}.
\]
\end{proposition}

\begin{proof}
Since \(\boldsymbol{x}\) is successful, there exists a feasible source--destination path under \(\boldsymbol{x}\). Removing a cycle cannot increase the lead time or cost and cannot decrease the bottleneck capacity because all lead times and costs are nonnegative. Repeatedly removing cycles and unnecessary arcs therefore produces a feasible minimal path \(P\). By Proposition~\ref{prop:path_feasibility},
\[
K(P,\boldsymbol{x})\geq\alpha(P).
\]
Consequently, the vector \(\boldsymbol{x}^{P}\) defined by
\eqref{eq:path_boundary_vector} satisfies
\[
\boldsymbol{x}^{P}\leq\boldsymbol{x}.
\]
The vector \(\boldsymbol{x}^{P}\) is itself successful. Since
\(\boldsymbol{x}\) is a lower boundary vector, it cannot strictly dominate another successful state. Therefore,
\(\boldsymbol{x}=\boldsymbol{x}^{P}\).
\end{proof}

Propositions~\ref{prop:path_to_boundary} and
\ref{prop:boundary_to_path} yield the following characterization:
\begin{equation}
\Psi^{\min}_{d,T,b}
=
\left\{
\boldsymbol{x}^{P}:
P\in\mathcal{P},\
L(P)<T,\
dC(P)\leq b,\
\alpha(P)\leq M(P)
\right\}.
\label{eq:boundary_characterization}
\end{equation}
Thus, determining the lower boundary set is equivalent to finding every potentially feasible minimal path and constructing its associated vector. This characterization forms the theoretical basis of the proposed algorithm. In contrast to approaches that first enumerate all minimal paths, the algorithm developed in Section~\ref{sec:algorithm} constructs only relevant paths and eliminates an infeasible partial path as soon as the new look-ahead conditions establish that it cannot be completed feasibly.

\subsection{Reliability Calculation from Boundary Vectors}
\label{subsec:reliability_boundary}

For every lower boundary vector
\(\boldsymbol{x}^{r}\in\Psi^{\min}_{d,T,b}\), define its upper-state event as
\begin{equation}
E_r
=
\left\{
\boldsymbol{x}:
\boldsymbol{x}\geq\boldsymbol{x}^{r}
\right\},
\qquad
r=1,2,\ldots,\sigma.
\label{eq:upper_state_event}
\end{equation}
By Lemma~\ref{lem:monotonicity}, every state belonging to \(E_r\) is successful. Conversely, every successful state dominates at least one lower boundary vector. Therefore,
\begin{equation}
\Psi_{d,T,b}
=
\bigcup_{r=1}^{\sigma}E_r
\label{eq:success_union}
\end{equation}
and
\begin{equation}
R_{d,T,b}
=
\Pr\left(
\bigcup_{r=1}^{\sigma}E_r
\right).
\label{eq:reliability_union}
\end{equation}

For arc \(a_i\), define the upper-tail probability
\begin{equation}
\overline{p}_i(k)
=
\Pr(x_i\geq k)
=
\sum_{q=k}^{M_i}p_{i,q}.
\label{eq:tail_probability}
\end{equation}
Under the independence assumption, the probability associated with a single lower boundary vector is
\begin{equation}
\Pr(E_r)
=
\prod_{i=1}^{m}
\overline{p}_i(x_i^{r}).
\label{eq:single_boundary_probability}
\end{equation}

The events \(E_1,\ldots,E_\sigma\) generally overlap. Therefore, their probabilities cannot simply be added. For any nonempty index set
\(S\subseteq\{1,\ldots,\sigma\}\), let
\begin{equation}
\boldsymbol{z}^{S}
=
\bigvee_{r\in S}\boldsymbol{x}^{r}.
\label{eq:intersection_vector}
\end{equation}
The intersection of the corresponding upper-state events is
\[
\bigcap_{r\in S}E_r
=
\left\{
\boldsymbol{x}:
\boldsymbol{x}\geq\boldsymbol{z}^{S}
\right\},
\]
and hence
\begin{equation}
\Pr\left(\bigcap_{r\in S}E_r\right)
=
\prod_{i=1}^{m}
\overline{p}_i(z_i^{S}).
\label{eq:intersection_probability}
\end{equation}
An exact inclusion--exclusion representation of the reliability is consequently
\begin{equation}
R_{d,T,b}
=
\sum_{\emptyset\neq S\subseteq\{1,\ldots,\sigma\}}
(-1)^{|S|+1}
\prod_{i=1}^{m}
\overline{p}_i(z_i^{S}).
\label{eq:inclusion_exclusion_reliability}
\end{equation}

Although \eqref{eq:inclusion_exclusion_reliability} provides a direct exact expression, its number of terms can grow exponentially with \(\sigma\). In practice, the union probability in
\eqref{eq:reliability_union} may instead be evaluated using an exact method such as the sum of disjoint products, inclusion--exclusion with preprocessing, a binary decision diagram, or a recursive factoring procedure
\cite{balan2003preprocessing,elkhadiri2023factoring}.

The principal objective of the present paper is the efficient exact generation of
\(\Psi^{\min}_{d,T,b}\). Once this set has been obtained, it constitutes a complete representation of the successful state space and can be supplied to any suitable exact reliability-evaluation method. 

\section{The Proposed Resource-Aware Look-Ahead Algorithm}
\label{sec:algorithm}

This section develops an exact sequential algorithm for generating all the \((d,T,b)\) lower boundary vectors characterized in
\eqref{eq:boundary_characterization}. The method retains the principal advantage of the node--child-matrix algorithm, namely, constructing relevant paths without receiving the complete set of minimal paths as input. Its main improvement is the use of destination-based information to determine whether a partial path can still be completed feasibly. The resulting look-ahead conditions account for the lead time, transmission cost, and attainable bottleneck capacity of the unconstructed portion of the path.

\subsection{Motivation and Limitation of the Existing NCM Algorithm}
\label{subsec:ncm_limitation}

The node--child-matrix (NCM) algorithm proposed in
\cite{forghani2023using} uses a backtracking procedure to construct simple source--destination paths. During the construction of a partial path, it maintains its accumulated lead time, accumulated unit transmission cost, and maximum attainable bottleneck capacity. After adding an arc, the algorithm verifies whether the resulting partial path can still satisfy the time, budget, and capacity conditions. A branch is terminated if one of these conditions is violated.

This strategy is exact and avoids generating all minimal paths in a separate preliminary stage. Nevertheless, its pruning conditions depend mainly on the resources already consumed by the partial path. They do not fully account for the resources that will necessarily be required to travel from the current node to the destination.

To explain this issue, let
\[
Q=(1,\ldots,v)
\]
be a partial simple path from the source to a current node \(v\). Denote its accumulated lead time and unit cost by
\begin{equation}
L(Q)=\sum_{a_i\in Q}\ell_i
\quad\text{and}\quad
C(Q)=\sum_{a_i\in Q}c_i,
\label{eq:partial_resources}
\end{equation}
respectively. Its maximum attainable bottleneck capacity is
\begin{equation}
U(Q)=\min_{a_i\in Q}M_i,
\label{eq:partial_bottleneck}
\end{equation}
where \(U(Q)=+\infty\) when \(Q\) contains no arc.

The existing NCM algorithm may continue expanding \(Q\) whenever
\[
L(Q)<T,\qquad dC(Q)\leq b,
\]
and its current bottleneck capacity is not already too small. However, these conditions do not guarantee that \(Q\) admits a feasible continuation. For example, every path from \(v\) to \(n\) may require so much additional lead time that the time limit will necessarily be exceeded. Similarly, every continuation may cause the budget constraint to be violated, or the arcs that provide a sufficiently short and inexpensive continuation may not have adequate maximum capacities.

Therefore, a partial path can pass all the tests of the existing NCM algorithm even though no feasible complete path can contain it. Such a partial path does not cause an incorrect solution, because it will eventually be rejected. It may, however, generate a large number of unnecessary descendant paths before its infeasibility becomes evident.

The proposed method addresses this limitation by calculating admissible lower bounds on the resources required to complete \(Q\). These values are used only as bounds; they do not select a particular continuation and do not assume that a shortest path is itself feasible. Consequently, the proposed tests never discard a partial path merely because one selected shortest path fails. A partial path is discarded only when the bounds prove that every possible continuation is infeasible.

For notational convenience, define
\begin{equation}
B=\frac{b}{d}
\label{eq:unit_path_budget}
\end{equation}
as the maximum allowable unit cost of a path. Thus, the budget condition \(dC(P)\leq b\) is equivalent to \(C(P)\leq B\).

\subsection{Basic Destination-Based Bounds}
\label{subsec:basic_bounds}

For every node \(v\in N\), let
\begin{equation}
\underline{L}(v)
=
\min\{L(R):R\text{ is a }v\text{--}n\text{ path}\}
\label{eq:basic_lead_bound}
\end{equation}
be the shortest lead time from \(v\) to the destination. Similarly, define
\begin{equation}
\underline{C}(v)
=
\min\{C(R):R\text{ is a }v\text{--}n\text{ path}\}
\label{eq:basic_cost_bound}
\end{equation}
as the least unit transmission cost from \(v\) to \(n\). If the destination is unreachable from \(v\), the corresponding value is \(+\infty\). At the destination,
\[
\underline{L}(n)=\underline{C}(n)=0.
\]

The two quantities may be attained by different paths. This does not affect their validity as lower bounds. For every continuation \(R\) from \(v\) to \(n\),
\begin{equation}
L(R)\geq\underline{L}(v)
\quad\text{and}\quad
C(R)\geq\underline{C}(v).
\label{eq:basic_bound_validity}
\end{equation}
Hence, if either
\begin{equation}
L(Q)+\underline{L}(v)\geq T
\label{eq:basic_lead_pruning}
\end{equation}
or
\begin{equation}
C(Q)+\underline{C}(v)>B
\label{eq:basic_cost_pruning}
\end{equation}
holds, the partial path \(Q\) cannot have a feasible continuation.

A further basic bound can be obtained from the attainable capacity. Let
\begin{equation}
W(v)
=
\max_{R:v\leadsto n}\min_{a_i\in R}M_i
\label{eq:widest_suffix}
\end{equation}
denote the widest-path capacity from \(v\) to \(n\). We use the conventions \(W(n)=+\infty\) and \(W(v)=0\) when \(n\) is unreachable from \(v\). Any completion of \(Q\) has a maximum bottleneck capacity no larger than
\begin{equation}
\widehat{U}(Q,v)
=
\min\{U(Q),W(v)\}.
\label{eq:maximum_completion_capacity}
\end{equation}
Therefore, the smallest transmission time that could optimistically be obtained by any completion of \(Q\) is bounded from below by
\begin{equation}
L(Q)+\underline{L}(v)
+
\left\lceil
\frac{d}{\widehat{U}(Q,v)}
\right\rceil.
\label{eq:basic_time_lower_bound}
\end{equation}

\begin{proposition}
\label{prop:basic_pruning}
A partial path \(Q=(1,\ldots,v)\) has no potentially
\((d,T,b)\)-feasible completion if at least one of the following conditions holds:
\begin{align}
\underline{L}(v)&=+\infty, \label{eq:basic_test_unreachable}\\
\widehat{U}(Q,v)&=0, \label{eq:basic_test_capacity}\\
C(Q)+\underline{C}(v)&>B, \label{eq:basic_test_cost}\\
L(Q)+\underline{L}(v)
+\left\lceil\frac{d}{\widehat{U}(Q,v)}\right\rceil
&>T. \label{eq:basic_test_time}
\end{align}
\end{proposition}

\begin{proof}
Let \(P=Q\oplus R\) be any completion of \(Q\), where \(R\) connects \(v\) to \(n\). From
\eqref{eq:basic_bound_validity},
\[
L(P)\geq L(Q)+\underline{L}(v)
\quad\text{and}\quad
C(P)\geq C(Q)+\underline{C}(v).
\]
Moreover,
\[
M(P)\leq\widehat{U}(Q,v).
\]
It follows that
\[
L(P)+\left\lceil\frac{d}{M(P)}\right\rceil
\geq
L(Q)+\underline{L}(v)
+\left\lceil\frac{d}{\widehat{U}(Q,v)}\right\rceil.
\]
Thus, any of
\eqref{eq:basic_test_unreachable}--\eqref{eq:basic_test_time}
makes every completion infeasible.
\end{proof}

These tests are inexpensive and can eliminate many partial paths. Nevertheless, they combine the unrestricted shortest lead time, least cost, and widest capacity, which may be achieved by three different continuations. They are therefore deliberately optimistic. The following subsections develop stronger bounds that incorporate the required capacity directly into the underlying networks.

\subsection{Capacity-Filtered Networks}
\label{subsec:filtered_networks}

Let
\begin{equation}
\Lambda
=
\{\lambda_1,\lambda_2,\ldots,\lambda_\kappa\}
=
\{M_i:M_i>0,\ a_i\in A\}
\label{eq:capacity_levels}
\end{equation}
be the set of distinct positive maximum-capacity levels, arranged so that
\[
\lambda_1<\lambda_2<\cdots<\lambda_\kappa.
\]
For each \(\lambda\in\Lambda\), define the capacity-filtered network
\begin{equation}
G_\lambda=(N,A_\lambda),
\qquad
A_\lambda=\{a_i\in A:M_i\geq\lambda\}.
\label{eq:filtered_network}
\end{equation}
Every path in \(G_\lambda\) can potentially provide a bottleneck capacity of at least \(\lambda\). As \(\lambda\) increases, the filtered network becomes sparser:
\[
A_{\lambda_1}\supseteq A_{\lambda_2}
\supseteq\cdots\supseteq A_{\lambda_\kappa}.
\]

For every \(v\in N\) and \(\lambda\in\Lambda\), define
\begin{align}
\underline{L}_{\lambda}(v)
&=
\min\{L(R):R\text{ is a }v\text{--}n
\text{ path in }G_\lambda\},
\label{eq:filtered_lead_bound}\\
\underline{C}_{\lambda}(v)
&=
\min\{C(R):R\text{ is a }v\text{--}n
\text{ path in }G_\lambda\}.
\label{eq:filtered_cost_bound}
\end{align}
Again, an unavailable path is represented by \(+\infty\), and
\[
\underline{L}_{\lambda}(n)
=
\underline{C}_{\lambda}(n)=0.
\]

The values in \eqref{eq:filtered_lead_bound} and
\eqref{eq:filtered_cost_bound} can be computed by applying a shortest-path algorithm from the destination in the reverse of each filtered network. Nonnegative lead times and costs allow Dijkstra's algorithm to be used. For an undirected network, the same computations are performed without reversing the arcs.

It is important to emphasize that
\(\underline{L}_{\lambda}(v)\) and
\(\underline{C}_{\lambda}(v)\) are not used to choose a specific suffix. They may correspond to different paths and are used only as optimistic lower bounds. Moreover, these paths may include a node already contained in \(Q\). Allowing this possibility can make the bounds weaker, but never invalid: forbidding previously visited nodes can only increase the resource requirements of an admissible simple continuation.

\subsection{Capacity-Aware Look-Ahead Results}
\label{subsec:capacity_aware_results}

Consider a partial path \(Q=(1,\ldots,v)\). If a complete path containing \(Q\) has bottleneck capacity \(\lambda\), then necessarily
\[
\lambda\leq U(Q),
\]
and its suffix from \(v\) to \(n\) must belong to \(G_\lambda\). These observations motivate the following set of admissible capacity levels:
\begin{equation}
\begin{split}
\mathcal{K}(Q,v)
=
\biggl\{
\lambda\in\Lambda:\;&
\lambda\leq U(Q),\
\underline{L}_{\lambda}(v)<+\infty,\
\underline{C}_{\lambda}(v)<+\infty,\\
&
L(Q)+\underline{L}_{\lambda}(v)
+\left\lceil\frac{d}{\lambda}\right\rceil\leq T,\\
&
C(Q)+\underline{C}_{\lambda}(v)\leq B
\biggr\}.
\end{split}
\label{eq:admissible_capacity_set}
\end{equation}

The time test in \eqref{eq:admissible_capacity_set} jointly considers the lead time already consumed, a lower bound on the remaining lead time under capacity requirement \(\lambda\), and the transmission term associated with that capacity. Similarly, the cost test combines the cost of the partial path with a lower bound on the cost of a continuation that can support capacity \(\lambda\).

\begin{lemma}
\label{lem:filtered_bounds}
Let \(P=Q\oplus R\) be a potentially feasible completion of \(Q\), and let
\[
\lambda=M(P)=\min_{a_i\in P}M_i.
\]
Then,
\begin{equation}
\underline{L}_{\lambda}(v)\leq L(R)
\quad\text{and}\quad
\underline{C}_{\lambda}(v)\leq C(R).
\label{eq:filtered_bound_admissibility}
\end{equation}
\end{lemma}

\begin{proof}
Every arc of \(P\), and hence every arc of \(R\), has maximum capacity at least \(\lambda\). Therefore, \(R\) is a path in \(G_\lambda\). The result follows immediately from the definitions of
\(\underline{L}_{\lambda}(v)\) and
\(\underline{C}_{\lambda}(v)\).
\end{proof}

\begin{theorem}
\label{thm:capacity_lookahead}
If
\[
\mathcal{K}(Q,v)=\varnothing,
\]
then no potentially \((d,T,b)\)-feasible path can contain \(Q\) as a prefix.
\end{theorem}

\begin{proof}
Suppose, to the contrary, that \(Q\) has a potentially feasible completion
\(P=Q\oplus R\). Let \(\lambda=M(P)\). Since the bottleneck of a path is attained by at least one of its arcs,
\(\lambda\in\Lambda\). Moreover,
\[
\lambda\leq U(Q).
\]
By Lemma~\ref{lem:filtered_bounds},
\[
L(Q)+\underline{L}_{\lambda}(v)
\leq L(P)
\]
and
\[
C(Q)+\underline{C}_{\lambda}(v)
\leq C(P)\leq B.
\]
Because \(P\) is potentially feasible,
\[
L(P)+\left\lceil\frac{d}{M(P)}\right\rceil
=
L(P)+\left\lceil\frac{d}{\lambda}\right\rceil
\leq T.
\]
Consequently,
\[
L(Q)+\underline{L}_{\lambda}(v)
+\left\lceil\frac{d}{\lambda}\right\rceil
\leq T.
\]
Thus, \(\lambda\in\mathcal{K}(Q,v)\), contradicting
\(\mathcal{K}(Q,v)=\varnothing\).
\end{proof}

The converse of Theorem~\ref{thm:capacity_lookahead} does not necessarily hold. In particular,
\(\mathcal{K}(Q,v)\neq\varnothing\) is a necessary but not sufficient condition for a feasible completion. The shortest-lead-time and least-cost suffixes for a given capacity level may be different, and either suffix may visit a node already present in \(Q\). This optimism is intentional: it ensures that the look-ahead rule is safe.

The basic and capacity-aware tests can be applied successively. The tests in Proposition~\ref{prop:basic_pruning} are evaluated first because they require only constant time after preprocessing. The more expensive search over
\(\mathcal{K}(Q,v)\) is performed only when the partial path passes the basic tests.

\subsection{Optional Cost--Time Coupled Bound}
\label{subsec:coupled_bound}

The filtered lead-time and cost bounds are computed separately. Therefore, a capacity level may remain in
\(\mathcal{K}(Q,v)\) even if no single suffix can simultaneously attain sufficiently small values of both resources. A stronger test can be obtained by forming a nonnegative weighted combination of lead time and cost.

Let \(\Theta\) be a finite set of nonnegative multipliers. For every
\(\lambda\in\Lambda\), \(\theta\in\Theta\), and \(v\in N\), define
\begin{equation}
H_{\lambda,\theta}(v)
=
\min_{R:v\leadsto n,\ R\subseteq G_\lambda}
\left\{
L(R)+\theta C(R)
\right\}.
\label{eq:coupled_shortest_path}
\end{equation}
This value is found by solving a shortest-path problem in \(G_\lambda\) with arc weights
\[
w_i(\theta)=\ell_i+\theta c_i.
\]

For a partial path \(Q\) and a capacity level \(\lambda\), define the residual lead-time and cost limits
\begin{align}
\tau_\lambda(Q)
&=
T-L(Q)-\left\lceil\frac{d}{\lambda}\right\rceil,
\label{eq:residual_lead_limit}\\
\beta(Q)
&=
B-C(Q).
\label{eq:residual_cost_limit}
\end{align}

\begin{lemma}
\label{lem:coupled_pruning}
If, for some \(\theta\in\Theta\),
\begin{equation}
H_{\lambda,\theta}(v)
>
\tau_\lambda(Q)+\theta\beta(Q),
\label{eq:coupled_pruning_condition}
\end{equation}
then no suffix in \(G_\lambda\) can complete \(Q\) within both the time and budget limits while supporting capacity \(\lambda\).
\end{lemma}

\begin{proof}
Suppose that a suffix \(R\) in \(G_\lambda\) satisfies
\[
L(R)\leq\tau_\lambda(Q)
\quad\text{and}\quad
C(R)\leq\beta(Q).
\]
Since \(\theta\geq0\),
\[
L(R)+\theta C(R)
\leq
\tau_\lambda(Q)+\theta\beta(Q).
\]
However, the definition of \(H_{\lambda,\theta}(v)\) implies
\[
H_{\lambda,\theta}(v)
\leq L(R)+\theta C(R),
\]
which contradicts \eqref{eq:coupled_pruning_condition}.
\end{proof}

Accordingly, the capacity set in
\eqref{eq:admissible_capacity_set} can optionally be strengthened as
\begin{equation}
\mathcal{K}^{+}(Q,v)
=
\left\{
\lambda\in\mathcal{K}(Q,v):
H_{\lambda,\theta}(v)
\leq
\tau_\lambda(Q)+\theta\beta(Q)
\text{ for every }\theta\in\Theta
\right\}.
\label{eq:strengthened_capacity_set}
\end{equation}
If \(\mathcal{K}^{+}(Q,v)=\varnothing\), the partial path can be safely discarded.

The set \(\Theta\) controls the tradeoff between preprocessing effort and pruning strength. The case \(\Theta=\varnothing\) gives the core algorithm based only on
\(\mathcal{K}(Q,v)\). A small set of multipliers can provide additional pruning at a moderate cost. Since the bound is valid for every nonnegative multiplier, the use or omission of this enhancement does not affect the correctness of the algorithm.

An exact resource-constrained shortest-path calculation could provide an even stronger coupled test. However, its computational cost may be considerable and could transfer too much effort from enumeration to preprocessing. The weighted bounds in
\eqref{eq:coupled_shortest_path} preserve the simplicity of ordinary shortest-path computations and are therefore more suitable as an optional enhancement.

\subsection{Network Preprocessing and Reduction}
\label{subsec:network_reduction}

Before beginning the path enumeration, several safe preprocessing operations can be performed.

First, nodes and arcs that do not belong to any source--destination path are removed by forward and backward reachability searches. For a directed network, a retained node must be reachable from node \(1\) and must also be able to reach node \(n\). An arc \((u,v)\) can be retained only if \(u\) is reachable from the source and the destination is reachable from \(v\).

Second, the values
\[
\underline{L}(v),\quad
\underline{C}(v), \text{ and}\quad 
W(v)
\]
are calculated for all nodes. If the initial partial path \(Q=(1)\) fails Proposition~\ref{prop:basic_pruning}, then the instance is infeasible and
\[
\Psi^{\min}_{d,T,b}=\varnothing.
\]

Third, the filtered bounds
\(\underline{L}_{\lambda}(v)\) and
\(\underline{C}_{\lambda}(v)\) are calculated for every
\(\lambda\in\Lambda\). If the optional coupled test is activated,
\(H_{\lambda,\theta}(v)\) is also calculated.

A further arc-reduction test can be applied by computing source-based filtered distances. For each \(\lambda\in\Lambda\), let
\begin{align}
\underline{L}^{s}_{\lambda}(u)
&=
\min\{L(Q):Q\text{ is a }1\text{--}u
\text{ path in }G_\lambda\},
\label{eq:source_lead_bound}\\
\underline{C}^{s}_{\lambda}(u)
&=
\min\{C(Q):Q\text{ is a }1\text{--}u
\text{ path in }G_\lambda\}.
\label{eq:source_cost_bound}
\end{align}

Consider a directed arc \(a_i=(u,v)\). Define
\begin{equation}
\begin{split}
\Lambda_i^{\mathrm{arc}}
=
\biggl\{
\lambda\in\Lambda:\;&
\lambda\leq M_i,\\
&
\underline{L}^{s}_{\lambda}(u)
+\ell_i+\underline{L}_{\lambda}(v)
+\left\lceil\frac{d}{\lambda}\right\rceil
\leq T,\\
&
\underline{C}^{s}_{\lambda}(u)
+c_i+\underline{C}_{\lambda}(v)
\leq B
\biggr\}.
\end{split}
\label{eq:arc_capacity_set}
\end{equation}

\begin{proposition}
\label{prop:arc_reduction}
If
\[
\Lambda_i^{\mathrm{arc}}=\varnothing,
\]
then arc \(a_i\) cannot belong to a potentially
\((d,T,b)\)-feasible path and may be removed.
\end{proposition}

\begin{proof}
Suppose that a potentially feasible path \(P\) contains \(a_i=(u,v)\), and let
\(\lambda=M(P)\). The prefix of \(P\) ending at \(u\), the arc \(a_i\), and the suffix beginning at \(v\) all belong to \(G_\lambda\). Their lead times and costs are therefore bounded from below by the corresponding source- and destination-based distances. Since \(P\) satisfies the time and budget restrictions, it follows that
\(\lambda\in\Lambda_i^{\mathrm{arc}}\), contradicting
\(\Lambda_i^{\mathrm{arc}}=\varnothing\).
\end{proof}

For an undirected arc \(\{u,v\}\), the test is performed in both orientations. The arc is removed only if it fails in both directions. The reduction is safe even if the paths producing the source- and destination-based bounds overlap or contain repeated nodes, because such paths only make the lower bounds more optimistic.

To keep the preprocessing cost controlled, the arc-reduction test can be performed once using distances calculated in the original filtered networks. Recomputing the distances after removing arcs and repeating the test until no additional reduction occurs is also exact, but its additional computational cost should be evaluated experimentally.

\subsection{The Proposed Algorithm}
\label{subsec:proposed_algorithm}

The proposed resource-aware look-ahead algorithm, abbreviated as RALA, consists of a preprocessing phase followed by a depth-first enumeration phase. It may use the NCM representation of
\cite{forghani2023using}; however, an arc-indexed adjacency list is preferable because it avoids the zero padding of the NCM and directly identifies the lead time, cost, and maximum capacity of every candidate arc. Both representations generate the same search tree.

The algorithm maintains the following information for each partial path \(Q=(1,\ldots,v)\):

\begin{itemize}
    \item the ordered node sequence of \(Q\);
    \item its arc set;
    \item its accumulated lead time \(L(Q)\);
    \item its accumulated unit cost \(C(Q)\);
    \item its current maximum bottleneck capacity \(U(Q)\); and
    \item a Boolean array indicating the nodes already contained in \(Q\).
\end{itemize}

The Boolean array prevents repeated nodes and guarantees that every constructed source--destination path is minimal. 

The algorithm can be used in three modes. RALA-B applies only the basic tests in Proposition~\ref{prop:basic_pruning}; RALA-C additionally applies the capacity-aware test in Theorem~\ref{thm:capacity_lookahead}; and the optional coupled mode further strengthens RALA-C by using \(\mathcal{K}^{+}(Q,v)\). For use in the pseudocode, define
\[
\operatorname{PassLevels}(Q,v)
=
\begin{cases}
\textnormal{true}, & \textnormal{for RALA-B},\\
\mathcal{K}(Q,v)\neq\varnothing, & \textnormal{for RALA-C},\\
\mathcal{K}^{+}(Q,v)\neq\varnothing, & \textnormal{when the coupled test is activated}.
\end{cases}
\label{eq:levels_function}
\]

The proposed method is presented in two parts for clarity. 
Algorithm~\ref{alg:rala} performs the preprocessing and initializes the
search, whereas Algorithm~\ref{alg:rala_extend} gives the recursive
procedure used to construct and examine the partial paths.

\begin{algorithm}[htbp]
\caption{Resource-aware look-ahead algorithm (RALA)}
\label{alg:rala}
\begin{algorithmic}[1]
\Require Network \(G\), demand \(d>0\), time limit \(T\), and budget \(b\)
\Ensure The set \(\Psi^{\min}_{d,T,b}\) of all lower boundary vectors

\State \(B\gets b/d\)
\State Remove nodes and arcs not belonging to any source--destination path
\State Compute \(\underline{L}(v)\), \(\underline{C}(v)\), and \(W(v)\) for all \(v\in N\)

\If{the initial path \(Q=(1)\) fails Proposition~\ref{prop:basic_pruning}}
    \State \Return \(\varnothing\)
\EndIf

\If{the capacity-aware test, coupled test, or arc-reduction phase is activated}
    \State Construct the set \(\Lambda\) of distinct positive capacity levels
    \ForAll{\(\lambda\in\Lambda\)}
        \State Construct \(G_\lambda\)
        \State Compute \(\underline{L}_{\lambda}(v)\) and
        \(\underline{C}_{\lambda}(v)\) for all \(v\in N\)
        \If{the coupled test is activated}
            \State Compute \(H_{\lambda,\theta}(v)\) for all
            \(v\in N\) and \(\theta\in\Theta\)
        \EndIf
    \EndFor
\EndIf

\If{the arc-reduction phase is activated}
    \State Compute the source-based filtered distances
    \State Remove every arc satisfying
    Proposition~\ref{prop:arc_reduction}
\EndIf

\State \(\Psi^{\min}_{d,T,b}\gets\varnothing\)
\State \(Q\gets(1)\), \(L_Q\gets0\), \(C_Q\gets0\), and \(U_Q\gets+\infty\)
\State Mark node \(1\) as visited
\State \Call{Extend}{$1,Q,L_Q,C_Q,U_Q$}
\State \Return \(\Psi^{\min}_{d,T,b}\)
\end{algorithmic}
\end{algorithm}


\begin{algorithm}[htbp]
\caption{Recursive extension procedure used by RALA}
\label{alg:rala_extend}
\begin{algorithmic}[1]
\Require Current node \(v\), partial path \(Q\), accumulated lead time \(L_Q\),
accumulated unit cost \(C_Q\), and bottleneck capacity \(U_Q\)
\Ensure The set \(\Psi^{\min}_{d,T,b}\) is updated whenever a new solution is found

\Procedure{Extend}{$v,Q,L_Q,C_Q,U_Q$}

    \If{\(v=n\)}
        \State \(\alpha\gets
        \left\lceil d/(T-L_Q)\right\rceil\)
        \If{\(L_Q<T\), \(C_Q\leq B\), and \(\alpha\leq U_Q\)}
            \State Construct \(\boldsymbol{x}^{Q}\) using
            \eqref{eq:path_boundary_vector}
            \State Add \(\boldsymbol{x}^{Q}\) to
            \(\Psi^{\min}_{d,T,b}\)
        \EndIf
        \State \Return
    \EndIf

    \If{\(Q\) fails Proposition~\ref{prop:basic_pruning}}
        \State \Return
    \EndIf

    \If{\(\operatorname{PassLevels}(Q,v)=\textnormal{false}\)}
        \State \Return
    \EndIf

    \ForAll{outgoing arcs \(a_i=(v,w)\)}
        \If{\(w\) is not visited}
            \State \(L'\gets L_Q+\ell_i\)
            \State \(C'\gets C_Q+c_i\)
            \State \(U'\gets\min\{U_Q,M_i\}\)

            \If{\(L'<T\), \(C'\leq B\), and \(U'>0\)}
                \State Append \(a_i\) and \(w\) to \(Q\)
                \State Mark \(w\) as visited
                \State \Call{Extend}{$w,Q,L',C',U'$}
                \State Mark \(w\) as unvisited
                \State Remove \(a_i\) and \(w\) from \(Q\)
            \EndIf
        \EndIf
    \EndFor

\EndProcedure
\end{algorithmic}
\end{algorithm}

In RALA-C and the coupled mode, Proposition~\ref{prop:basic_pruning} should be applied before scanning the capacity levels. Furthermore, the capacity levels larger than \(U_Q\) can be skipped immediately. If the levels are stored in increasing order, the scan can stop after reaching the first value exceeding \(U_Q\). RALA-B performs no capacity-level scan.

The algorithm stores \(\boldsymbol{x}^{Q}\) as a full \(m\)-component state vector only when a solution is found. During the search, maintaining the path arcs, \(L_Q\), \(C_Q\), and \(U_Q\) is sufficient. The bottleneck value can be updated in constant time when extending the path. Since recursion automatically restores the previous value during backtracking, it is unnecessary to recalculate the bottleneck over the entire partial path.

\subsection{Correctness Results}
\label{subsec:correctness}

We first establish that the algorithm never returns an unsuitable vector.

\begin{theorem}
\label{thm:rala_soundness}
Every state vector generated by any mode of RALA (Algorithms~\ref{alg:rala} and~\ref{alg:rala_extend})
is a \((d,T,b)\) lower boundary vector.
\end{theorem}

\begin{proof}
A vector is generated only after the search reaches node \(n\). The visited-node condition ensures that the resulting source--destination path \(P\) is simple and hence minimal. Before constructing the vector, the algorithm verifies
\[
L(P)<T,\qquad C(P)\leq B,
\quad\text{and}\quad
\alpha(P)\leq U(P)=M(P).
\]
Therefore, \(P\) is potentially feasible according to
\eqref{eq:potential_path_feasibility}. Proposition~\ref{prop:path_to_boundary} then implies that the vector constructed through
\eqref{eq:path_boundary_vector} is a lower boundary vector.
\end{proof}

The next result establishes completeness.

\begin{theorem}
\label{thm:rala_completeness}
Every mode of RALA generates every vector in \(\Psi^{\min}_{d,T,b}\).
\end{theorem}

\begin{proof}
Let
\(\boldsymbol{x}\in\Psi^{\min}_{d,T,b}\). By
Proposition~\ref{prop:boundary_to_path}, there exists a potentially feasible minimal path \(P\) such that
\(\boldsymbol{x}=\boldsymbol{x}^{P}\).

The depth-first procedure considers every simple source--destination path unless one of its prefixes is pruned. Consider an arbitrary prefix \(Q=(1,\ldots,v)\) of \(P\). Since \(P\) is a potentially feasible completion of \(Q\), Proposition~\ref{prop:basic_pruning} cannot discard \(Q\). This establishes the safety of every pruning step performed by RALA-B. If RALA-C is used, Theorem~\ref{thm:capacity_lookahead} implies
\[
\mathcal{K}(Q,v)\neq\varnothing.
\]
If the optional coupled bound is also used, Lemma~\ref{lem:coupled_pruning} guarantees that the capacity level \(M(P)\) satisfies every coupled inequality and therefore belongs to \(\mathcal{K}^{+}(Q,v)\). Consequently, \(\operatorname{PassLevels}(Q,v)\) is true in every mode.
Finally, if the optional arc-reduction phase is activated, Proposition~\ref{prop:arc_reduction} ensures that no arc of \(P\) is removed during preprocessing.

Consequently, no prefix of \(P\) is pruned, and the depth-first search eventually reaches the destination through \(P\). At that point, the algorithm constructs
\(\boldsymbol{x}^{P}=\boldsymbol{x}\).
\end{proof}

\begin{theorem}
\label{thm:rala_no_duplicates}
No mode of RALA generates duplicate lower boundary vectors.
\end{theorem}

\begin{proof}
The recursive search fixes the order of the arcs from the source to the destination, and each outgoing arc is examined once whenever its tail is the current node. Therefore, every simple path is generated at most once.

Moreover, \(d>0\) implies \(\alpha(P)\geq1\). Hence, the positive components of
\(\boldsymbol{x}^{P}\) identify exactly the arcs of \(P\). Two distinct minimal paths have different arc sets and thus cannot produce the same state vector through
\eqref{eq:path_boundary_vector}. Therefore, no duplicate vector is generated.
\end{proof}

Theorems~\ref{thm:rala_soundness}--\ref{thm:rala_no_duplicates} establish the principal correctness result.

\begin{theorem}
\label{thm:rala_correctness}
Every mode of RALA generates the complete set \(\Psi^{\min}_{d,T,b}\) exactly and without duplicates.
\end{theorem}

\subsection{Complexity Analysis}
\label{subsec:complexity}

Let
\[
\kappa=|\Lambda|\leq m
\]
be the number of distinct positive maximum-capacity levels, and let
\[
\rho=|\Theta|
\]
be the number of multipliers used by the optional coupled bound. Let \(\Delta\) denote the maximum out-degree of the network.

The initial reachability searches require
\[
O(n+m)
\]
time. The basic shortest-lead-time, least-cost, and widest-path values can be computed in
\[
O(m+n\log n)
\]
time using adjacency lists and heap-based implementations~\cite{ahuja1993network, fredman1987fibonacci}.

Let \(Z_B\) denote the number of partial paths for which the recursive procedure is invoked by RALA-B. Since RALA-B performs no capacity-level preprocessing or scanning, its total running time is
\begin{equation}
O\bigl(m+n\log n+Z_B\Delta+\sigma m\bigr),
\label{eq:rala_b_total_complexity}
\end{equation}
where the last term accounts for explicitly constructing the \(\sigma\) output vectors.

For each capacity level, RALA-C solves two destination-based shortest-path problems, one using lead times and the other using costs. Therefore, its capacity-filtered preprocessing requires
\begin{equation}
O\bigl(
\kappa(m+n\log n)
\bigr)
\label{eq:core_preprocessing_complexity}
\end{equation}
time. If source-based distances are also calculated for the arc-reduction phase, the number of shortest-path computations is doubled, but the asymptotic order in
\eqref{eq:core_preprocessing_complexity} remains unchanged.

For each pair \((\lambda,\theta)\), the optional coupled bound requires one additional shortest-path computation. Its preprocessing time is consequently
\begin{equation}
O\bigl(
\rho\kappa(m+n\log n)
\bigr).
\label{eq:coupled_preprocessing_complexity}
\end{equation}
Testing all arcs for the reduction condition requires
\[
O\bigl(m\kappa(1+\rho)\bigr)
\]
time once the required distances have been calculated.

The enumeration phase is necessarily output-sensitive. For RALA-C and the optional coupled mode, let \(Z_C\) be the number of partial paths for which the recursive procedure is invoked, including paths subsequently eliminated by a look-ahead test. At each such state, at most \(\Delta\) outgoing arcs are examined. A direct scan of all capacity levels, including the optional coupled tests, requires
\[
O\bigl(\kappa(1+\rho)\bigr)
\]
time. If \(\sigma\) boundary vectors are generated, explicitly constructing their \(m\)-component representations requires \(O(\sigma m)\) time. The total running time is therefore
\begin{equation}
O\Bigl(
\kappa(1+\rho)(m+n\log n)
+
m\kappa(1+\rho)
+
Z_C\bigl[\Delta+\kappa(1+\rho)\bigr]
+
\sigma m
\Bigr).
\label{eq:total_complexity}
\end{equation}
After absorbing the arc-screening term into preprocessing, this expression can be written more compactly as
\begin{equation}
O\Bigl(
\kappa(1+\rho)(m+n\log n)
+
Z_C\bigl[\Delta+\kappa(1+\rho)\bigr]
+
\sigma m
\Bigr).
\label{eq:compact_total_complexity}
\end{equation}

For RALA-C without the optional coupled bound, \(\rho=0\), and the complexity becomes
\begin{equation}
O\Bigl(
\kappa(m+n\log n)
+
Z_C(\Delta+\kappa)
+
\sigma m
\Bigr).
\label{eq:core_total_complexity}
\end{equation}

The quantities \(Z_B\) and \(Z_C\) may be exponential in the worst case because a network can contain exponentially many simple source--destination paths and lower boundary vectors. This behavior cannot be avoided by any exact enumeration algorithm when the output itself is exponential. The purpose of the proposed bounds is to reduce these quantities substantially in practical instances by rejecting an infeasible partial path before enumerating its descendants.

Excluding the stored output vectors, RALA-B requires \(O(m+n)\) working memory for the network representation, visited-node array, basic destination-based bounds, and recursion stack. The memory required by RALA-C to store the filtered lead-time and cost bounds is
\[
O(\kappa n).
\]
The optional coupled values require an additional
\[
O(\rho\kappa n)
\]
memory. For RALA-C and the optional coupled mode, the network representation, visited-node array, and recursion stack require
\[
O(m+n)
\]
space, since a simple path contains at most \(n\) nodes. If all output vectors are retained explicitly, their storage requires
\[
O(\sigma m).
\]
Thus, the total memory requirement is
\begin{equation}
O\bigl(
m+n+\kappa n(1+\rho)+\sigma m
\bigr).
\label{eq:memory_complexity}
\end{equation}
If the boundary vectors are processed or written sequentially rather than retained in memory, the final term can be omitted from the working-memory requirement.

The worst-case order alone does not fully represent the expected advantage of RALA because the existing NCM algorithm and both RALA variants are exponential path-enumeration procedures. Their practical difference lies mainly in the number of partial paths examined. Therefore, the computational experiments report the corresponding search counts, the number of extensions rejected by each pruning rule, preprocessing time, enumeration time, and total running time. These measurements determine whether the reduction in the search tree compensates for the additional preprocessing and per-node look-ahead calculations.

\section{Illustrative Example}
\label{sec:example}

This section illustrates the principal steps of RALA on a small benchmark network. The example shows how the capacity-filtered look-ahead information is calculated, how it is used during path construction, and how the resulting lower boundary vectors can be employed to calculate the exact reliability. To highlight the effect of the new pruning conditions, the optional cost--time coupled bound and the arc-reduction phase are not activated in this example.

\subsection{Benchmark Network and Input Data}
\label{subsec:example_network}

Consider the five-node undirected multistate flow network depicted in Figure~\ref{fig:illustrative_network}. This topology was also considered in \cite{forghani2023using}. Node \(1\) is the source, and node \(5\) is the destination. Although the network is represented as undirected, a constructed path begins at node \(1\), terminates upon reaching node \(5\), and cannot visit any node more than once.
The maximum capacity, lead time, and unit transmission cost of each arc are given by 
\(\boldsymbol{M}=(4,2,5,4,3,3,4,5),\ \boldsymbol{\ell}=(3,1,2,3,4,3,2,3),\) and \(\boldsymbol{c}=(2,3,4,3,2,3,2,1)\).
Suppose that the objective is to transmit
\(
d=7
\)
units of flow from node \(1\) to node \(5\) within
\(
T=9
\)
units of time and at a total cost not exceeding
\(
b=63.
\)
The corresponding maximum unit path cost is
\(
B=\frac{b}{d}=\frac{63}{7}=9.
\)
Therefore, a potentially feasible path \(P\) must satisfy
\(
C(P)\leq9
\)
and
\[
L(P)+\left\lceil\frac{7}{M(P)}\right\rceil\leq9.
\]

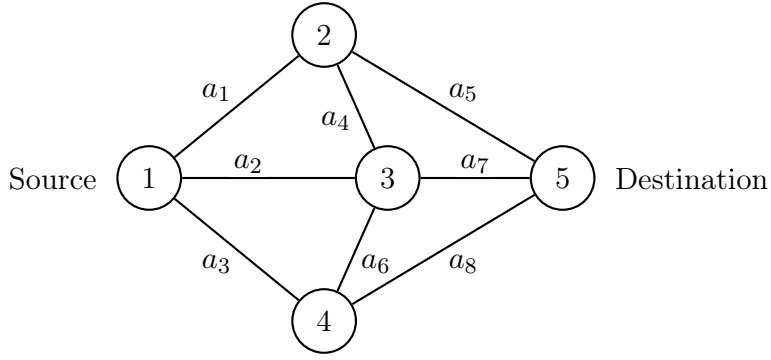
\begin{figure}[tp]
\centering
\begin{tikzpicture}[
    scale=1.05,
    every node/.style={transform shape},
    vertex/.style={
        circle,
        draw,
        thick,
        minimum size=8mm,
        inner sep=0pt,
        font=\small
    },
    edgeLabel/.style={
        fill=white,
        inner sep=1.5pt,
        font=\small
    }
]

\node[vertex] (v1) at (0,0) {1};
\node[vertex] (v2) at (2.2,1.8) {2};
\node[vertex] (v3) at (3.0,0) {3};
\node[vertex] (v4) at (2.2,-1.8) {4};
\node[vertex] (v5) at (5.2,0) {5};

\draw[thick] (v1) -- node[edgeLabel,above left] {$a_1$} (v2);
\draw[thick] (v1) -- node[edgeLabel,above left] {$a_2$} (v3);
\draw[thick] (v1) -- node[edgeLabel,below left] {$a_3$} (v4);
\draw[thick] (v2) -- node[edgeLabel,below left] {$a_4$} (v3);
\draw[thick] (v2) -- node[edgeLabel,above right] {$a_5$} (v5);
\draw[thick] (v3) -- node[edgeLabel,below right] {$a_6$} (v4);
\draw[thick] (v3) -- node[edgeLabel,above] {$a_7$} (v5);
\draw[thick] (v4) -- node[edgeLabel,below right] {$a_8$} (v5);

\node[font=\small,left=1mm of v1] {Source};
\node[font=\small,right=1mm of v5] {Destination};

\end{tikzpicture}
\caption{The five-node benchmark network used in the illustrative example.}
\label{fig:illustrative_network}
\end{figure}


For the reliability calculation, assume that the arc capacities are mutually independent and that each capacity level of a given arc has the same probability. That is,
\begin{equation}
\Pr(X_i=k)=\frac{1}{M_i+1},
\qquad
k=0,1,\ldots,M_i.
\label{eq:example_capacity_distribution}
\end{equation}
The uniform distributions in \eqref{eq:example_capacity_distribution} are used only to keep the final probability calculation transparent. The boundary-vector generation phase does not depend on these probabilities.

\subsection{Precomputed Look-Ahead Tables}
\label{subsec:example_tables}

The distinct positive maximum-capacity levels are
\(
\Lambda=\{2,3,4,5\}.
\)
For every \(\lambda\in\Lambda\), the filtered network \(G_\lambda\) contains only the arcs whose maximum capacities are at least \(\lambda\). For example,
\(
A_4=\{a_1,a_3,a_4,a_7,a_8\},
\)
whereas
\(
A_5=\{a_3,a_8\}.
\)

Applying the reverse shortest-path calculations described in
Subsection~\ref{subsec:filtered_networks} gives the filtered lead-time and cost bounds reported in Table~\ref{tab:example_filtered_bounds}. An infinite entry indicates that the destination cannot be reached from the corresponding node in that filtered network.

\begin{table}[htbp]
\centering
\caption{Capacity-filtered destination-based bounds for the illustrative network.}
\label{tab:example_filtered_bounds}
\resizebox{\textwidth}{!}{%
\begin{tabular}{c cc cc cc cc}
\toprule
& \multicolumn{2}{c}{\(\lambda=2\)}
& \multicolumn{2}{c}{\(\lambda=3\)}
& \multicolumn{2}{c}{\(\lambda=4\)}
& \multicolumn{2}{c}{\(\lambda=5\)}\\
\cmidrule(lr){2-3}
\cmidrule(lr){4-5}
\cmidrule(lr){6-7}
\cmidrule(lr){8-9}
Node \(v\)
& \(\underline{L}_{2}(v)\) & \(\underline{C}_{2}(v)\)
& \(\underline{L}_{3}(v)\) & \(\underline{C}_{3}(v)\)
& \(\underline{L}_{4}(v)\) & \(\underline{C}_{4}(v)\)
& \(\underline{L}_{5}(v)\) & \(\underline{C}_{5}(v)\)\\
\midrule
1 & 3 & 4 & 5 & 4 & 5 & 5 & 5 & 5\\
2 & 4 & 2 & 4 & 2 & 6 & 4 & \(\infty\) & \(\infty\)\\
3 & 2 & 2 & 2 & 2 & 2 & 2 & \(\infty\) & \(\infty\)\\
4 & 3 & 1 & 3 & 1 & 3 & 1 & 3 & 1\\
5 & 0 & 0 & 0 & 0 & 0 & 0 & 0 & 0\\
\bottomrule
\end{tabular}%
}
\end{table}

Some entries in Table~\ref{tab:example_filtered_bounds} are attained by the following paths:

\begin{itemize}
    \item For \(\lambda=2\),
    \(\underline{L}_{2}(1)=3\) is attained by
    \(1\text{--}3\text{--}5\), whereas
    \(\underline{C}_{2}(1)=4\) is attained by
    \(1\text{--}2\text{--}5\).

    \item For \(\lambda=3\), arc \(a_2\) is removed, and
    \(\underline{L}_{3}(1)=5\) is attained by
    \(1\text{--}4\text{--}5\).

    \item For \(\lambda=4\), arcs \(a_2\), \(a_5\), and \(a_6\) are unavailable. From node \(2\), the shortest admissible suffix is
    \(2\text{--}4\text{--}5\), with lead time \(6\) and unit cost \(4\).

    \item For \(\lambda=5\), only \(a_3\) and \(a_8\) remain. Consequently, nodes \(2\) and \(3\) cannot reach the destination in \(G_5\).
\end{itemize}

These observations also illustrate why the shortest lead-time and least-cost values should be interpreted as bounds rather than as a selected continuation. For instance,
\(\underline{L}_{2}(1)\) and \(\underline{C}_{2}(1)\) are obtained from different paths. Nevertheless, both values are valid lower bounds on the corresponding resources required by any complete path.

At the initial partial path \(Q=(1)\), we have
\(
L(Q)=C(Q)=0,\text{ and } U(Q)=+\infty.
\)
Using \eqref{eq:admissible_capacity_set}, the four capacity levels satisfy
\begin{align*}
\lambda=2:\quad&
0+3+\left\lceil\frac{7}{2}\right\rceil
=7\leq9,
&
0+4&\leq9,\\
\lambda=3:\quad&
0+5+\left\lceil\frac{7}{3}\right\rceil
=8\leq9,
&
0+4&\leq9,\\
\lambda=4:\quad&
0+5+\left\lceil\frac{7}{4}\right\rceil
=7\leq9,
&
0+5&\leq9,\\
\lambda=5:\quad&
0+5+\left\lceil\frac{7}{5}\right\rceil
=7\leq9,
&
0+5&\leq9.
\end{align*}
Hence,
\(
\mathcal{K}((1),1)=\{2,3,4,5\},
\)
and the search begins.

\subsection{Step-by-Step Execution}
\label{subsec:example_execution}

Suppose that the outgoing arcs are examined in their index order. The principal partial paths investigated by RALA are summarized in Table~\ref{tab:example_execution}. The notation ``solution'' means that the destination has been reached and the corresponding boundary vector is generated.

\begin{table}[htbp]
\centering
\caption{Principal steps of RALA for the illustrative network.}
\label{tab:example_execution}
\resizebox{\textwidth}{!}{%
\begin{tabular}{ccccccl}
\toprule
Partial path \(Q\) & \(L(Q)\) & \(C(Q)\) & \(U(Q)\) &
\(\mathcal{K}(Q,v)\) & Current node & Decision\\
\midrule
\((1)\)             & 0 & 0 & \(\infty\) & \(\{2,3,4,5\}\) & 1 & Extend\\
\((1,2)\)           & 3 & 2 & 4 & \(\varnothing\) & 2 & Prune\\
\((1,3)\)           & 1 & 3 & 2 & \(\{2\}\) & 3 & Extend\\
\((1,3,4)\)         & 4 & 6 & 2 & \(\varnothing\) & 4 & Prune\\
\((1,3,5)\)         & 3 & 5 & 2 & -- & 5 & First solution\\
\((1,4)\)           & 2 & 4 & 5 & \(\{2,3,4,5\}\) & 4 & Extend\\
\((1,4,2)\)         & 5 & 7 & 4 & \(\varnothing\) & 2 & Prune\\
\((1,4,3)\)         & 5 & 7 & 3 & \(\varnothing\) & 3 & Prune\\
\((1,4,5)\)         & 5 & 5 & 5 & -- & 5 & Second solution\\
\bottomrule
\end{tabular}%
}
\end{table}

The first extension adds arc \(a_1\), giving
\[
Q=(1,2),\qquad
L(Q)=3,\qquad C(Q)=2,\qquad U(Q)=4.
\]
This partial path illustrates the additional strength of the capacity-aware test. In fact, the basic destination-based test does not eliminate it. Since
\[
\underline{L}(2)=4,\qquad
\underline{C}(2)=2,\qquad
W(2)=4,
\]
we obtain
\[
L(Q)+\underline{L}(2)
+\left\lceil\frac{d}{\min\{U(Q),W(2)\}}\right\rceil
=
3+4+\left\lceil\frac{7}{4}\right\rceil
=9,
\]
and
\[
C(Q)+\underline{C}(2)=2+2=4\leq9.
\]
Thus, the partial path passes the basic tests.

The capacity-filtered tests, however, give
\begin{align*}
\lambda=2:\quad&
3+\underline{L}_{2}(2)
+\left\lceil\frac{7}{2}\right\rceil
=3+4+4=11>9,\\
\lambda=3:\quad&
3+\underline{L}_{3}(2)
+\left\lceil\frac{7}{3}\right\rceil
=3+4+3=10>9,\\
\lambda=4:\quad&
3+\underline{L}_{4}(2)
+\left\lceil\frac{7}{4}\right\rceil
=3+6+2=11>9.
\end{align*}
The level \(\lambda=5\) is unavailable because
\(\lambda>U(Q)=4\). Therefore,
\[
\mathcal{K}((1,2),2)=\varnothing,
\]
and the entire branch beginning with \(1\text{--}2\) is discarded.

Notice that this conclusion is not obtained by examining only one shortest continuation. Instead, the calculations prove that no continuation can simultaneously provide a suitable lead time and the required bottleneck capacity. In particular, the ordinary widest-path and shortest-lead-time information is overly optimistic in this case because the suffix attaining the widest capacity does not attain the shortest lead time.

After backtracking to the source, the algorithm constructs
\[
Q=(1,3).
\]
For this partial path,
\[
L(Q)=1,\qquad C(Q)=3,\qquad U(Q)=2.
\]
Only \(\lambda=2\) needs to be considered, and
\[
1+\underline{L}_{2}(3)
+\left\lceil\frac{7}{2}\right\rceil
=
1+2+4=7\leq9,
\]
while
\[
3+\underline{C}_{2}(3)=3+2=5\leq9.
\]
Thus,
\[
\mathcal{K}((1,3),3)=\{2\}.
\]
Extending this partial path to node \(4\) gives
\[
Q=(1,3,4),\qquad
L(Q)=4,\qquad C(Q)=6,\qquad U(Q)=2.
\]
Its best possible capacity-filtered completion would require at least
\[
4+\underline{L}_{2}(4)
+\left\lceil\frac{7}{2}\right\rceil
=
4+3+4=11>9.
\]
This branch is therefore discarded.

The next extension reaches the destination through
\[
P^{1}=(1,3,5)=\{a_2,a_7\}.
\]
For this path,
\[
L(P^{1})=1+2=3,\qquad
C(P^{1})=3+2=5,
\]
and
\[
M(P^{1})=\min\{2,4\}=2.
\]
Its minimum required capacity is
\[
\alpha(P^{1})
=
\left\lceil
\frac{7}{9-3}
\right\rceil
=2.
\]
Since
\[
\alpha(P^{1})\leq M(P^{1})
\quad\text{and}\quad
dC(P^{1})=7(5)=35\leq63,
\]
the first lower boundary vector is generated.

The algorithm then examines the branch beginning with
\(a_3\). For
\[
Q=(1,4),
\]
we have
\[
L(Q)=2,\qquad C(Q)=4,\qquad U(Q)=5,
\]
and all four capacity levels remain admissible:
\[
\mathcal{K}((1,4),4)=\{2,3,4,5\}.
\]
The continuations through nodes \(2\) and \(3\) are eliminated by the look-ahead tests, as reported in Table~\ref{tab:example_execution}. The direct continuation through \(a_8\) reaches the destination and produces
\[
P^{2}=(1,4,5)=\{a_3,a_8\}.
\]
For this path,
\[
L(P^{2})=2+3=5,\qquad
C(P^{2})=4+1=5,
\]
and
\[
M(P^{2})=\min\{5,5\}=5.
\]
The required capacity is
\[
\alpha(P^{2})
=
\left\lceil
\frac{7}{9-5}
\right\rceil
=2.
\]
Therefore, \(P^{2}\) also produces a lower boundary vector.

The network contains nine simple source--destination paths. Nevertheless, RALA reaches the destination through only the two potentially feasible ones. Four partial paths are eliminated before reaching node \(5\), and their descendant paths are not constructed.

\subsection{Final State Vectors and Reliability}
\label{subsec:example_results}

Applying \eqref{eq:path_boundary_vector} to \(P^{1}\) and \(P^{2}\) gives
\begin{align}
\boldsymbol{x}^{1}
&=
(0,2,0,0,0,0,2,0),
\label{eq:example_boundary_1}\\
\boldsymbol{x}^{2}
&=
(0,0,2,0,0,0,0,2).
\label{eq:example_boundary_2}
\end{align}
Hence,
\begin{equation}
\Psi^{\min}_{7,9,63}
=
\left\{
(0,2,0,0,0,0,2,0),
(0,0,2,0,0,0,0,2)
\right\}.
\label{eq:example_boundary_set}
\end{equation}

The first vector indicates that arcs \(a_2\) and \(a_7\) must each have a capacity of at least \(2\). Under this state,
\[
f(7,\boldsymbol{x}^{1},P^{1})
=
3+\left\lceil\frac{7}{2}\right\rceil
=7\leq9.
\]
Likewise, the second vector requires capacities of at least \(2\) on \(a_3\) and \(a_8\), and
\[
f(7,\boldsymbol{x}^{2},P^{2})
=
5+\left\lceil\frac{7}{2}\right\rceil
=9.
\]

Let
\[
E_1=\{\boldsymbol{X}\geq\boldsymbol{x}^{1}\}
\quad\text{and}\quad
E_2=\{\boldsymbol{X}\geq\boldsymbol{x}^{2}\}.
\]
From the uniform distributions in
\eqref{eq:example_capacity_distribution},
\begin{align*}
\Pr(X_2\geq2)
&=\frac{1}{3},&
\Pr(X_7\geq2)
&=\frac{3}{5},\\
\Pr(X_3\geq2)
&=\frac{4}{6}=\frac{2}{3},&
\Pr(X_8\geq2)
&=\frac{4}{6}=\frac{2}{3}.
\end{align*}
Consequently,
\begin{align}
\Pr(E_1)
&=
\Pr(X_2\geq2)\Pr(X_7\geq2)
=
\frac{1}{3}\frac{3}{5}
=
\frac{1}{5},
\label{eq:example_event_1}\\
\Pr(E_2)
&=
\Pr(X_3\geq2)\Pr(X_8\geq2)
=
\frac{2}{3}\frac{2}{3}
=
\frac{4}{9}.
\label{eq:example_event_2}
\end{align}

The two events involve disjoint sets of arcs. Under the independence assumption,
\begin{equation}
\Pr(E_1\cap E_2)
=
\Pr(E_1)\Pr(E_2)
=
\frac{1}{5}\frac{4}{9}
=
\frac{4}{45}.
\label{eq:example_intersection}
\end{equation}
Therefore, the exact quickest path reliability is
\begin{align}
R_{7,9,63}
&=
\Pr(E_1\cup E_2)\notag\\
&=
\Pr(E_1)+\Pr(E_2)-\Pr(E_1\cap E_2)\notag\\
&=
\frac{1}{5}+\frac{4}{9}-\frac{4}{45}\notag\\
&=
\frac{5}{9}\notag\\
&\approx0.555556.
\label{eq:example_reliability}
\end{align}

This example demonstrates that the proposed look-ahead procedure preserves all feasible paths while eliminating partial paths whose infeasibility can be established from the resources required to reach the destination. In particular, the branch beginning with \(1\text{--}2\) passes the unrestricted destination-based test but is rejected by the stronger capacity-filtered test. This distinction illustrates the main advantage of the proposed algorithm over a procedure based only on the resources already consumed by a partial path.

\section{Concluding Remarks}
\label{sec:conclusion}

This paper developed a resource-aware look-ahead approach for generating the complete set of lower boundary vectors associated with the quickest path reliability problem under simultaneous time and budget constraints. The approach extends the node--child-matrix search by considering not only the resources already consumed by a partial path but also the resources that will necessarily be required to reach the destination. Shortest-lead-time, least-cost, and widest-capacity information is computed for the remaining network and used to recognize partial paths that cannot lead to a feasible solution.

Two exact variants of the resulting resource-aware look-ahead algorithm were presented. RALA-B applies the basic destination-based lead-time and cost bounds, whereas RALA-C strengthens these tests using capacity-filtered subnetworks. The latter variant accounts more directly for the interaction between the bottleneck capacity of a path, its transmission time, and its cost. A cost--time coupled condition and an iterative arc-reduction procedure were also introduced as optional mechanisms for obtaining additional pruning before or during enumeration.

The theoretical results establish that all the proposed elimination conditions are safe. Whenever a partial path is removed, no extension of that path can satisfy the required demand, time limit, and budget simultaneously. The resulting algorithms therefore generate every $(d,T,b)$ lower boundary vector and do not produce duplicate solutions. Moreover, they operate without requiring the complete collection of minimal paths as an input. The output-sensitive complexity analysis distinguishes the costs associated with preprocessing, partial-path examination, and construction of the final state vectors, thereby clarifying the sources of computational effort in both variants.

The present work concentrates on the exact generation of lower boundary vectors. Once these vectors have been obtained, the quickest path reliability can be evaluated using any suitable exact union-probability method. This separation allows improvements in boundary-vector generation to be incorporated without changing the subsequent reliability-calculation procedure. At the same time, the overall difficulty of an instance continues to depend on the number and structure of the partial paths explored and on the number of lower boundary vectors that must be represented explicitly. The relative benefit of the capacity-aware conditions may also depend on the number of distinct capacity levels and on the tightness of the demand, time, and budget requirements.

The proposed framework suggests several directions for further research. Partial paths belonging to different branches of the search tree can be examined independently, making the method naturally suitable for parallel execution. In addition, the repeated evaluation of resource bounds across multiple partial paths may be accelerated through vectorized representations and batched operations. Parallel and vectorized implementations, combined with suitable workload-balancing strategies, may therefore extend the exact method to larger and more demanding networks. Another promising direction is the development of an adaptive procedure that selects the appropriate look-ahead conditions according to their estimated pruning benefit and evaluation cost. Extensions to related multistate-network models involving additional resource constraints, multiple sources or destinations, and alternative transmission requirements also merit investigation.

\section*{Acknowledgments}
This research received partial support from the São Paulo Research Foundation (FAPESP) under grant 2023/13667-5 and from the National Council for Scientific and Technological Development (CNPq) under grant 303906/2025-1. 

\bibliographystyle{elsarticle-num} 
\bibliography{majidbibliography_updated}

\end{document}